\documentclass[english]{ccdconf}

\usepackage[english]{babel}
\newtheorem{assumption}{\textbf{Assumption}}
\newtheorem{lemma}{\textbf{Lemma}}
\newtheorem{definition}{\textbf{Definition}}
\newtheorem{theorem}{\textbf{Theorem}}

\newtheorem{remark}{\textbf{Remark}}

\newtheorem{problem}{\textbf{Problem}}

\usepackage{graphicx}
\usepackage{epsfig}
\IfFileExists{newtxmath.sty}
{
	\usepackage{newtxtext,newtxmath}
}{
	\usepackage{mathptmx}
	\usepackage{times}
}
\usepackage{amsmath}
\usepackage{amsfonts}
\usepackage{amssymb}
\usepackage{mathrsfs}
\usepackage{xcolor}
\usepackage{graphicx}

\newcommand{\T}{^{\mbox{\tiny T}}}
\newcommand{\R}{\mathbb{R}}
\newcommand{\C}{\mathbb{C}}

\newcommand{\eps}{\varepsilon}
\let\leq\leqslant
\let\geq\geqslant

\newenvironment{proof}[1][Proof]%
{\par\noindent\textit{#1:\ }}%
{\hspace*{\fill} \rule{6pt}{6pt}}
\newenvironment{proof*}[1][Proof]%
{\par\noindent\textit{#1:\ }}{}

\DeclareMathOperator{\diag}{diag}

\DeclareMathOperator{\rank}{rank}

\usepackage{tikz}
\usetikzlibrary{shapes,calc,arrows,patterns,decorations.pathmorphing
	,decorations.markings}
\usetikzlibrary{arrows.meta}

\newenvironment{system}[1]%
{\setlength{\arraycolsep}{0.5mm}\left\{ \; \begin{array}{#1}}%
	{\end{array} \right.}
\newenvironment{system*}[1]%
{\setlength{\arraycolsep}{0.5mm} \begin{array}{#1}}%
	{\end{array}}

\begin{document}
	\title{Delayed state synchronization of continuous-time multi-agent systems in the presence of unknown communication delays}
	\author{Zhenwei Liu\aref{neu}, Ali Saberi\aref{wsu},
		Anton A. Stoorvogel\aref{ut}, Rong Li\aref{shfc}}
	
	\affiliation[neu]{College of Information Science and
		Engineering, Northeastern University, Shenyang 110819, China
		\email{liuzhenwei@ise.neu.edu.cn}}
	\affiliation[wsu]{School of Electrical Engineering and Computer
		Science, Washington State University, Pullman, WA 99164, USA
		\email{saberi@eecs.wsu.edu}}
	\affiliation[ut]{Department of Electrical Engineering,
		Mathematics and Computer Science, University of Twente, Enschede, The Netherlands
		\email{A.A.Stoorvogel@utwente.nl}}
	\affiliation[shfc]{School of Statistics and Mathematics, Shanghai Lixin University of Accounting and Finance, Shanghai 201209, China
		\email{lirong@shfc.edu.cn}}
	
	\maketitle
	
	\begin{abstract}
		This paper studies delayed synchronization of continuous-time
		multi-agent systems (MAS) in the presence of unknown nonuniform
		communication delays. A \emph{delay-free} transformation is
		developed based on a communication network which is a directed
		spanning tree, which can transform the original MAS to a new one
		without delays. By using this transformation, we design a static
		protocol for full-state coupling and a dynamic protocol for delayed
		state synchronization for homogeneous MAS via full- and
		partial-state coupling. Meanwhile, the delayed output synchronization
		is also studied for heterogeneous MAS, which is achieved by using a
		low-gain and output regulation based dynamic protocol design via the
		delay-free transformation.
	\end{abstract}
	

	\section{Introduction}
	
	The problem of synchronization among agents in a multi-agent system
	has received substantial attention in recent years, because of its
	potential applications in cooperative control of autonomous vehicles,
	distributed sensor network, swarming and flocking and others. The
	objective of synchronization is to secure an asymptotic agreement on a
	common state or output trajectory through decentralized control
	protocols (see \cite{bai-arcak-wen,lewis-zhang-hengster-movric-das,liu-saberi-stoorvogel-zhang-ijrnc,liu-zhang-saberi-stoorvogel-auto,mesbahi-egerstedt,ren-book,wu-book} and references therein).

	Recently synchronization in a network with time delay has attracted a
	great deal of interest. As clarified in \cite{cao-yu-ren-chen}, we can
	identify two kinds of delay. Firstly there is \emph{communication
		delay}, which results from limitations on the communication between
	agents. Secondly we have \emph{input delay}, which is due to
	computational limitations of an individual agent. Many works have
	focused on dealing with input delay (see e.g.\ \cite{liu-stoorvogel-saberi-zhang-cdc17,liu-zhang-saberi-stoorvogel-ejc,saber-murray2,stoorvogel-saberi-acc,tian,aalst,xiao-wang,zhang-saberi-stoorvogel-cdc15}), but communication delay is much less understood at this moment. In the case of communication delay, only for
	a constant synchronization trajectory do we preserve the
	diffusive nature of the network. This diffusive nature is an
	intrinsic part of the currently available design techniques and
	hence only this case has been studied. Some works in this area can be seen in 
	\cite{ferrari-trecate,lin-jia-auto,tian-liu,tian}. 
	
	References \cite{chopra-tac} and \cite{chopra-spong-cdc06} solved the
	synchronization problem for nonlinear heterogeneous MAS with unknown
	non-uniform constant communication delays. Some other results for non-uniform communication delays can also be found in \cite{lin-jia-auto,munz-papachristodoulou-allgower,
		munz-papachristodoulou-allgower2,tian-liu,xiao-wang-tac}.
	
	On the other hand, the concept of delayed synchronization was introduced in \cite{chopra-tac} and \cite{chopra-spong-cdc06}. Compared with general synchronization problem, delayed synchronization allows a fixed signal lag from parents node to their son node when constant communication delay is considered. It means that there exists a fixed distance (or some other physical quantities) between two agents to keep moving. But, due to the restriction of strongly connected network, the synchronized trajectory must converse to a constant.
	
	In this paper, we study delayed synchronization problems of MAS in the
	presence of unknown communication delays. The communication network is
	assumed to be a directed spanning tree (i.e., it has one root node and
	the other non-root nodes have indegree one). The contribution of this
	paper is threefold:
	\begin{itemize}
		\item A \emph{delay-free} transformation is established to remove the
		effect from unknown communication delays, and obtain a transformed
		MAS without communication delays.
		\item We develop the delayed state synchronization results of
		homogeneous MAS based on the delay-free transformation, and obtain a
		dynamic synchronized trajectory. Static and dynamic protocol designs
		are provided for both full- and partial-state coupling cases
		respectively.
		\item We also develop the delayed output synchronization results of
		heterogeneous MAS. A low-gain and output regulation based dynamic
		protocol design is provided via the delay-free transformation.
	\end{itemize}
	
	We will see that, compared to earlier work, our approach is limited to
	a graph which is a directed spanning tree. However, the intrinsic
	advantage is that the synchronized trajectory is not limited to a
	constant but will follow the trajectory of the root agent.
	
	{\bf Notations and definitions}:	
	Given a matrix $A\in \mathbb{R}^{m\times n}$, $A\T$ and $A^{*}$ denote
	the transpose and conjugate transpose of $A$, respectively while
	$\|A\|$ denotes the induced 2-norm of $A$. A square matrix $A$ is said
	to be Hurwitz stable if all its eigenvalues are in the open left half
	complex plane. $A\otimes B$ depicts the Kronecker product between $A$
	and $B$. $I_n$ denotes the $n$-dimensional identity matrix and $0_n$
	denotes $n\times n$ zero matrix; we will use $I$ or $0$ if the
	dimension is clear from the context.
	
	A \emph{weighted directed graph} $\mathcal{G}$ is defined by a triple
	$(\mathcal{V}, \mathcal{E}, \mathcal{A})$ where
	$\mathcal{V}=\{1,\ldots, N\}$ is a node set, $\mathcal{E}$ is a set of
	pairs of nodes indicating connections among nodes, and
	$\mathcal{A}=[a_{ij}]\in \mathbb{R}^{N\times N}$ is the weighting
	matrix, and $a_{ij}>0$ iff $(j,i)\in \mathcal{E}$. Each pair in
	$\mathcal{E}$ is called an \emph{edge}. A \emph{path} from node $i_k$
	to $i_1$ is a sequence of nodes $\{i_1,\ldots, i_k\}$ such that
	$(i_{j}, i_{j+1})\in \mathcal{E}$ for $j=1,\ldots, k-1$.  A
	\emph{directed tree} is a subgraph (subset of nodes and edges) in
	which every node has exactly one parent node except for one node,
	called the \emph{root}, which has no parent node. In this case, the
	root has a directed path to every other node in the tree.  A
	\emph{directed spanning tree} is a subgraph which is a directed tree
	containing all the nodes of the original graph.  An agent is called a
	\emph{root agent} if it is the root of some directed spanning tree of
	the associated graph.  Let $\Pi_{\mathcal{G}}$ denote the set of all
	root agents for a graph.  For a weighted graph $\mathcal{G}$, a matrix
	$L=[\ell_{ij}]$ with
	\begin{equation}\label{Lij}
	\ell_{ij}=
	\begin{system}{cl}
	\sum_{k=1}^{N} a_{ik}, & i=j,\\
	-a_{ij}, & i\neq j,
	\end{system}
	\end{equation}
	is called the \emph{Laplacian matrix} associated with the graph
	$\mathcal{G}$. In the case where $\mathcal{G}$ has non-negative
	weights, $L$ has all its eigenvalues in the closed right half plane
	and at least one eigenvalue at zero associated with right eigenvector
	$\textbf{1}$.

	\section{Problem formulation}
	
	We will study a MAS consisting of $N$ identical agents:
	\begin{equation}\label{eq1}
	\begin{cases}
	\dot{x}_i(t)=Ax_i(t)+Bu_i(t),\\
	y_i(t)=Cx_i(t),
	\end{cases}
	\end{equation}
	where $x_i(t)\in\mathbb{R}^n$, $u_i(t)\in\mathbb{R}^m$ and
	$y_i(t)\in\mathbb{R}^p$ are the state, input and the output,
	respectively, of agent $i$ for $i=1,\ldots, N$.
	
	\begin{assumption}\label{assume}
		We assume that
		\begin{itemize}
			\item ($A$, $B$, $C$) is stabilizable and detectable.
			\item All eigenvalues of $A$ are in the closed left half complex
			plane.
		\end{itemize}
	\end{assumption}
	
	The communication network provides agent $i$ with the following
	information,
	\begin{equation}\label{eq2}
	\zeta_i(t)=\sum_{j=1}^{N}a_{ij}\left[ y_i(t)-y_j(t-\tau_{ij}) \right]
	\end{equation}
	where $a_{ij}\geq 0$ and $a_{ii}=0$. This communication topology of
	the network can be described by a weighted graph $\mathcal{G}$ with
	weighting matrix $\mathcal{A}=[a_{ij}]$. We can obtain the associated
	Laplacian matrix $L$ via \eqref{Lij}.
	
	Here $\tau_{ij}\in\mathbb{R}^+$ represents an unknown constant
	communication delays from agent $j$ to agent $i$. This communication
	delay implies that it takes $\tau_{ij}$ seconds for agent $j$ to
	transfer its state information to agent $i$. Furthermore, we assume Agent 1 is root of graph in this paper.

	\begin{definition}\label{ungrN}
		For any $\beta>0$, let $\mathbb{G}_{\beta}^N$ denote the set of
		directed graphs with $N$ nodes which are equal to a directed
		spanning tree for which the corresponding Laplacian matrix $L$ is
		lower triangular with the top row identical to zero which has the
		property that $\ell_{ii}\geq\beta$ for $i=2,\ldots,N$ while agent
		$1$ is the root agent. Similarly for any $\alpha>\beta>0$, let
		$\mathbb{G}_{\alpha,\beta}^N$ denote the set of directed graphs with
		$N$ nodes which are equal to a directed spanning tree for which the
		corresponding Laplacian matrix $L$ is lower triangular with the
		first row equal to zero with the property that
		$\beta\leq\ell_{ii}\leq\alpha$ for $i=2,\ldots,N$.
	\end{definition}
	
	\begin{remark}
		Note that any graph which is a directed spanning tree will 
		have a lower triangular Laplacian matrix after a possible reordering
		of the agents.
	\end{remark}
	
	For the graph defined by Definition \ref{ungrN}, we know the
	Laplacian matrix $L$ is lower triangular with the first row identical
	to zero. Therefore, we have
	\[
	L=\begin{pmatrix}
	0 & 0 & 0 & \cdots & 0 \\
	\ell_{21} & \ell_{22} & 0 & \cdots & 0 \\
	\ell_{31} & \ell_{32} & \ell_{33} & \ddots & \vdots \\
	\vdots & \ddots & \ddots & \ddots & 0\\
	\ell_{N1} & \cdots & \ell_{N,N-2} & \ell_{N,N-1} & \ell_{N,N}
	\end{pmatrix}.
	\]
	Since the graph is equal to a directed spanning tree, there are in
	every row (except the first one) exactly two elements unequal to $0$.
	
	Our goal is to achieve delayed state synchronization among
	agents in a MAS, that is
	\begin{equation}\label{synchro}
	\lim_{t\to \infty}\, \left[ x_i(t)-x_j(t-\tau_{ij})\right] = 0,
	\end{equation}
	for all $i,j\in\{1,\ldots,N\}$.

	For the MAS \eqref{eq1}, we formulate delayed state synchronization
	problems as follows.
	
	\begin{problem}\label{prob1}
		Consider a MAS described by agents \eqref{eq1} and \eqref{eq2}
		associated with a directed graph $\mathcal{G}\in\mathbb{G}_{\beta}^N$ is equal to a spanning
		tree, where $\mathbb{G}_{\beta}^N$
		is defined in Definition \ref{ungrN}. The \emph{delayed state
			synchronization problem} with a set of graphs $\mathbb{G}_{\beta}^N$ in the
		presence of unknown, nonuniform, arbitrarily large communication
		delays is to find a distributed static protocol of the type,
		\begin{equation}\label{hocde-state1}  
		u_i(t)=F\zeta_i(t),\quad (i=1,\ldots,N)
		\end{equation}
		for each agent such that \eqref{synchro} is satisfied for all
		$i,j\in\{1,\ldots,N\}$, for any directed graph
		$\mathcal{G}\in \mathbb{G}_\beta^N$ and for any communication delay
		$\tau_{ij}\in\R^+$.
	\end{problem}
	
	\begin{problem}\label{prob2}
		Consider a MAS described by agents \eqref{eq1} and \eqref{eq2}
		associated with a directed graph $\mathcal{G}\in\mathbb{G}_{\alpha,\beta}^N$ is equal to a spanning
		tree, where $\mathbb{G}_{\alpha,\beta}^N$
		is defined in Definition \ref{ungrN}. The \emph{delayed state
			synchronization problem} with a set of graphs $\mathbb{G}_{\alpha,\beta}^N$ in the
		presence of unknown, nonuniform, arbitrarily large communication
		delays is to find a distributed dynamic protocol of the type,
		\begin{equation}\label{hocde-state2}  
		\begin{system}{cl}
		\dot{\chi}_i(t)&=A_c\chi_i(t)+B_c\zeta_i(t),\\
		u_i(t)&=C_c\chi_i(t)+D_c\zeta_i(t),
		\end{system}
		\end{equation}
		for each agent such that \eqref{synchro} is satisfied for all
		$i,j\in\{1,\ldots,N\}$, for any directed graph
		$\mathcal{G}\in \mathbb{G}_\beta^N$ and for any communication delay
		$\tau_{ij}\in\R^+$.
	\end{problem}
	
	\section{Delayed state synchronization for homogeneous MAS with
		communication delays}
	
	In this section, we will give
	delayed state synchronization results based on algebraic Riccati
	equation for full- and partial-state coupling.
	
	\subsection{Full-state coupling}
	
	Firstly, we consider full-state coupling (i.e., $C=I$). We
	define
	\[
	\tilde{x}_i(t)=x_i(t+\bar{\tau}_{i,1})
	\] 
	where $\bar{\tau}_{i,1}$ denotes the sum of delays from agent $i$ to
	the root (agent $1$) based on its path, and
	$\tau_{ij}=\bar{\tau}_{i,1}-\bar{\tau}_{j,1}$.  Then, we have
	\begin{align*}
	\tilde{\zeta}_i(t)&=\zeta_i(t+\bar{\tau}_{i,1})=\sum_{j=1}^{N}a_{ij}\left[ x_i(t+\bar{\tau}_{i,1})-x_j(t+\bar{\tau}_{i,1}-\tau_{ij})\right]\\
	&=\sum_{j=1}^{N}a_{ij}(\tilde{x}_i(t)-\tilde{x}_j(t))
	\end{align*}	
	and $\tilde{u}_i(t)=u_i(t+\bar{\tau}_{i,1})$. 
	Especially, we have $\bar{\tau}_{1,1}=0$, i.e., $\tilde{x}_1(t)=x_1(t)$
	and $\tilde{u}_1(t)=0$ (since $\tilde{\zeta}_1(t)=\zeta_1(t)=0$). Thus, \eqref{eq1}, \eqref{eq2} and \eqref{hocde-state1} yield:
	\begin{equation}\label{sdp-ts1}
	\dot{\tilde{x}}_i(t)=A{\tilde{x}}_i(t)+\sum_{j=1}^{N}\ell_{ij}BF\tilde{x}_j(t)
	\end{equation}
	Thus, our synchronization objective can be expressed as
	\begin{equation}\label{synchron}
	\lim_{t\to \infty} \left[ \tilde{x}_i(t)-\tilde{x}_j(t)\right]=0.
	\end{equation}
	We define
	\[
	{\tilde{x}}(t)=\begin{pmatrix}
	{\tilde{x}}_1(t) \\ \tilde{x}_2(t) \\ \vdots \\ \tilde{x}_N(t)
	\end{pmatrix}.
	\]
	We have 
	\begin{equation}\label{sdp-ts2}
	\dot{\tilde{x}}(t)=\left(I\otimes A + L\otimes (BF)\right)\tilde{x}(t).
	\end{equation}
	This is referred to as our \emph{delay-free transformation}.
	
	Further, let 
	\[
	\eta_1(t)=\tilde{x}_1(t), \text{ and } 
	\eta_i(t)=\tilde{x}_i(t)-\tilde{x}_1(t) \text{ with } i=2,\ldots,N,
	\]
	we have
	\[
	\eta=\begin{pmatrix}
	\eta_1 \\ \eta_2 \\ \vdots \\ \eta_N
	\end{pmatrix}=\begin{pmatrix}
	\tilde{x}_1(t) \\ \tilde{x}_2(t)-\tilde{x}_1(t) \\
	\vdots\\\tilde{x}_N(t)-\tilde{x}_1(t)  
	\end{pmatrix}=(T\otimes I) \tilde{x}(t).
	\]
	for some suitable defined matrix $T$.  Thus, based on the above
	transformation $T$, we obtain new expression of \eqref{sdp-ts2},
	\begin{equation}\label{sdp-ts3}
	\dot{\eta}(t)=\left(I_N\otimes A +L_Q\otimes (BF)\right)\eta(t)
	\end{equation}
	where 
	\begin{align}
	L_Q=TLT^{-1} &= \begin{pmatrix}
	0 & 0 \\ 0 &  L_Q
	\end{pmatrix} \nonumber \\
	&=\begin{pmatrix}
	0 & 0 & 0 & \cdots & 0\\
	0 & \ell_{22} & 0 & \cdots & 0\\
	0 & \ell_{32} & \ell_{33} & \ddots & \vdots \\
	\vdots & \vdots & \ddots & \ddots & 0 \\
	0 & \ell_{N2} & \cdots & \ell_{N,N-1} & \ell_{N,N}
	\end{pmatrix},\label{L_Q}
	\end{align}
	where $L_Q$ is a positive-definite lower triangular matrix.
	
	Obviously, due to the structure of $L_Q$, the synchronization of
	\eqref{sdp-ts2} is equivalent to the asymptotic stability of the
	following $N-1$ subsystems,
	\begin{equation}\label{sdp-ts4}
	\dot{\eta}_i(t)=(A+\ell_{ii}BF)\eta_i(t), \, \ i=2,\ldots,N.
	\end{equation}
	If \eqref{sdp-ts4} is globally asymptotically
	stable for $i=2,\ldots,N$, we see from the above that
	$\eta_i(t)\rightarrow 0$ for $i=2,\ldots,N$. This implies that 
	\[
	\tilde{x}(t) - (T^{-1}\otimes I)\begin{pmatrix} \eta_1(t) \\ 0 \\
	\vdots \\ 0 \end{pmatrix} \rightarrow 0.
	\]
	Note that the first column of $T^{-1}$ is equal to
	the  vector $\mathbf{1}$ and therefore
	\[
	\tilde{x}_i(t) - \eta_1(t) \rightarrow 0
	\]
	for $i=1,\ldots,N$. This implies that we achieve state synchronization.
	
	Conversely, suppose that the network \eqref{sdp-ts4} reaches state
	synchronization. In this case, we shall have
	\[
	\tilde{x}(t) - \mathbf{1}\otimes \tilde{x}_1(t) \to 0
	\]  
	for all initial conditions. Then
	$\eta(t)- (T\mathbf{1})\otimes \tilde{x}_1(t) \to 0$. Since $\mathbf{1}$
	is the first column of $T$, we have
	\[
	T\mathbf{1}=\begin{pmatrix}
	1 \\ 0 \\ \vdots \\ 0
	\end{pmatrix}.
	\]
	Therefore, $\eta(t) - (T\mathbf{1})\otimes \tilde{x}_1(t) \to 0$ implies
	that $\eta_1(t) - \tilde{x}_1(t) \to 0$ and $\eta_i(t)\to 0$ for $i=2,\ldots,N$
	for all initial conditions.
	
	Thus, we obtain the following lemma.
	
	\begin{lemma}\label{thm-tra}
		The MAS \eqref{sdp-ts2} achieves state synchronization if and only
		if the system \eqref{sdp-ts4} is globally asymptotically stable for
		$i=2,3,\ldots,N$. The synchronized trajectory converses to the
		trajectory of the root agent.
	\end{lemma}
	
	We let
	\[
	w=\begin{pmatrix} 
	1 & 0 & \cdots & 0 
	\end{pmatrix}
	\]
	be the normalized left eigenvector associated with the zero eigenvalue
	of $L$. Then, from Lemma \ref{thm-tra}, we have
	\begin{equation}\label{sdq-ts5}
	\dot{\eta}_1(t) = A \eta_1(t), \quad \eta_1(0)= \tilde{x}_1(0),
	\end{equation}
	In other words, $\eta_1(0)$ is the initial condition of the root
	agent. Therefore, the synchronized trajectory given by \eqref{sdq-ts5}
	yields that the synchronized trajectory is given by
	\begin{equation}\label{hoct-syn-traj-full}
	x_s(t)=e^{At}\tilde{x}_1(0)=e^{At}x_1(0),
	\end{equation}
	which is the trajectory of the root agent and delay-free. Therefore,
	the root agent is sometimes referred to as the leader.

	\subparagraph*{Protocol design:}
	For full-state coupling,
	we design a parameterized static protocol of the form:
	\begin{equation}\label{hoct-comp3-pole-full1}
	u_i(t)=-\rho B\T P\zeta_i(t),
	\end{equation}
	where $P>0$ is the unique solution of the continuous-time algebraic
	Riccati equation,
	\begin{equation}\label{hoct-are2}
	A\T P+PA-PBB\T P+ Q=0,
	\end{equation}
	with $Q>0$, and $\rho\geq \frac{1}{2\beta}$ with
	$\ell_{ii}\geq\beta>0$.
	
	Using an algebraic Riccati equation we can design a suitable protocol
	provided $(A,B)$ is stabilizable. The synchronization based on
	protocol \eqref{hoct-comp3-pole-full1} is as follows.
	
	\begin{theorem}\label{hoct-theorem-are}
		Consider a MAS described by \eqref{eq1} and \eqref{eq2}. Let any
		$\beta>0$ be given, and consider the set of network graphs
		$\mathbb{G}_{\beta}^N$ with $\ell_{ii}\geq\beta$.
		
		If $(A,B)$ is stabilizable, then the state synchronization problem
		stated in Problem \ref{prob1} with $\mathbf{G}=\mathbb{G}_{\beta}^N$
		is solvable. In particular, the protocol
		\eqref{hoct-comp3-pole-full1} solves the state synchronization
		problem for any graph $\mathcal{G}\in\mathbb{G}_{\beta}^N$ and
		$\tau_{ij}\in \R^+$. Moreover, the synchronized trajectory is given
		by \eqref{hoct-syn-traj-full}.
	\end{theorem}
	
	\begin{proof}
		For protocol \eqref{hoct-comp3-pole-full1}, we can obtain
		\[
		\tilde{u}_i(t)=-\rho B\T P\tilde{\zeta}_i(t).
		\]
		By using the delay-free transformation, it means that we only need
		to prove that the system
		\begin{equation}
		\dot{z}(t)=(A-\ell_{ii} \rho BB\T P)z(t)
		\end{equation}
		is asymptotically stable for any $\ell_{ii}$ that satisfies
		$\ell_{ii}\geq \beta$.
		
		We observe that
		\begin{align*}
		(A-\ell_{ii} &\rho BB\T P)\T P+P(A-\ell_{ii} \rho BB\T P)\\
		&=-Q-(2\ell_{ii}\rho-1)PBB\T P\\
		&\leq -Q.
		\end{align*}
		Therefore, $(A-\ell_{ii} \rho BB\T P)$ is Hurwitz stable for any
		$\ell_{ii}\geq \beta>0$. Based on Lemma \ref{thm-tra}, the delayed
		state synchronization result can be proved.
	\end{proof}

	\subsection{Partial-state coupling}
	
	In the following, we give the transformation for partial
	state coupling.  Similar to the case of full-state coupling, we have
	the following expression for \eqref{eq1} and \eqref{eq2} by using
	our delay-free transformation,
	\begin{equation}\label{dft-psc1}
	\begin{system}{ll}
	\dot{\tilde{x}}_i(t)=A{\tilde{x}}_i(t)+B\tilde{u}_i(t)\\
	\tilde{\zeta}_i(t)=\sum_{j=1}^{N}
	a_{ij}(\tilde{y}_i(t)-\tilde{y}_j(t)) 
	\end{system}
	\end{equation}
	with $\tilde{x}_i(t)=x_i(t-\bar{\tau}_{i,1})$,
	$\tilde{y}_i(t)=y_i(t-\bar{\tau}_{i,1})$. 
	
	The MAS described by \eqref{eq1} and \eqref{eq2} after implementing
	the dynamic protocol \eqref{hocde-state2} is described by
	\begin{equation}\label{hoct-sys1}
	\begin{system}{ccl}
	\dot{\bar{x}}_i(t) &=&
	{\setlength{\arraycolsep}{2mm}\begin{pmatrix}
		A & BC_c \\ 0 & A_c
		\end{pmatrix}}\bar{x}_i(t)+
	{\setlength{\arraycolsep}{2mm}\begin{pmatrix}
		BD_c \\ B_c
		\end{pmatrix}}\tilde{\zeta}_i(t),\\
	\tilde{y}_i(t) &=&
	{\setlength{\arraycolsep}{2mm}\begin{pmatrix} C & 0\end{pmatrix}}
	\bar{x}_i(t),\\ 
	\tilde{\zeta}_i(t)&=& {\displaystyle \sum_{j=1}^N
		a_{ij}(\tilde{y}_i(t)-\tilde{y}_j(t)),} 
	\end{system}
	\end{equation}
	for $i=1,\ldots,N$, where
	\[
	\bar{x}_i(t)= \begin{pmatrix} \tilde{x}_i(t) \\
	\chi_i(t) \end{pmatrix},
	\qquad \bar{x}(t) = \begin{pmatrix} \bar{x}_1(t) \\ \vdots \\
	\bar{x}_N(t) \end{pmatrix}
	\]
	Define
	\[
	\bar{A}=\begin{pmatrix}
	A & BC_c \\ 0 & A_c
	\end{pmatrix},
	\bar{B}=\begin{pmatrix}
	BD_c \\ B_c
	\end{pmatrix}, 
	\bar{C}=\begin{pmatrix}
	C & 0
	\end{pmatrix}.
	\]
	Then, the overall dynamics of the $N$ agents can be written as
	\begin{equation} \label{hoct-sys1-partial}
	\dot{\bar{x}}(t)= (I_N \otimes \bar{A}+L\otimes \bar{B}\bar{C})\bar{x}(t).
	\end{equation}
	So, this is the delay-free system obtained after our transformation
	for MAS with unknown communication delays via partial state coupling.
	
	The synchronization of \eqref{hoct-sys1-partial} is equivalent to the
	asymptotic stability of the following $N-1$ subsystems,
	\begin{equation}\label{dft-psc4}
	\dot{\bar{\eta}}_i=(\bar{A}+\ell_{ii}\bar{B}\bar{C})\bar{\eta}_i, \,
	\ i=2,\ldots,N. 
	\end{equation}
	
	Similar to Lemma \ref{thm-tra}, we obtain the following lemma for
	partial state coupling.
	
	\begin{lemma}\label{thm-tra-part}
		The MAS \eqref{hoct-sys1-partial} achieves state synchronization if
		and only if the system \eqref{dft-psc4} is globally asymptotically
		stable for $i=2,3,\ldots,N$. The synchronized trajectory converges
		to the trajectory of the root agent of \eqref{hoct-sys1-partial}.
	\end{lemma}
	
	Meanwhile,  we have 
	\begin{equation}\label{dft-psc5}
	\dot{\bar{\eta}}_1(t) = A \bar{\eta}_1(t),\qquad \bar{\eta}_1(0)= \bar{x}_1(0),
	\end{equation}
	Therefore, the synchronized trajectory given by \eqref{dft-psc5}
	yields that the synchronized trajectory is given by
	\begin{equation}\label{dft-psc6}
	x_s(t)=(I\quad 0)e^{\bar{A}t}\bar{x}_1(0)=e^{At}{x}_1(0).
	\end{equation}

	\subparagraph*{Protocol design:}
	For partial-state coupling, we design a parameterized dynamic
	protocol of the form:
	\begin{equation}\label{hoct-comp3-pole-partial}
	\begin{system}{ccl}
	\dot{\chi}_i(t) &=& (A+KC)\chi_i(t)-K\zeta_i(t) ,\\
	u_i(t) &=& -{\beta}^{-1} B\T {P}_\delta\chi_i(t),
	\end{system}
	\end{equation}
	where $K$ is a matrix such that $A+KC$ is Hurwitz stable,
	${P}_\delta>0$ is the unique solution of the continuous-time algebraic
	Riccati equation,
	\begin{equation}\label{are234r}
	A\T{P}_\delta+{P}_\delta A-{P}_\delta
	BB\T {P}_\delta+ \delta I=0. 
	\end{equation}
	with $\ell_{ii}\geq\beta>0$.
	
	\begin{theorem}\label{hoct-theorem-are-part}
		Consider a MAS described by \eqref{eq1} and \eqref{eq2}. Let any
		$\alpha>\beta>0$ be given, and consider the set of network graphs
		$\mathbb{G}_{\alpha,\beta}^N$ with $\alpha\geq\ell_{ii}\geq\beta$.
		
		If $(A,B)$ is stabilizable and $(A,C)$ is observable, then the state
		synchronization problem stated in Problem \ref{prob2} with
		$\mathbf{G}=\mathbb{G}_{\alpha,\beta}^N$ is solvable. In particular,
		there exists a $\delta^{*}>0$ such that for any
		$\delta\in(0,\delta^{*}]$, the dynamic protocol
		\eqref{hoct-comp3-pole-partial} solves the state synchronization
		problem for any graph $\mathcal{G}\in\mathbb{G}_{\alpha,\beta}^N$
		and $\tau_{ij}\in \R^+$. Moreover, the synchronized trajectory is
		given by \eqref{dft-psc6}.
	\end{theorem}
	
	\begin{proof}
		For dynamic protocol \eqref{hoct-comp3-pole-partial}, we have
		\[
		\begin{system}{ccl}
		\dot{\tilde{\chi}}_i (t)&=& (A+KC)\tilde{\chi}_i(t)-K\tilde{\zeta}_i(t) ,\\
		\tilde{u}_i(t) &=& -{\beta}^{-1} B\T {P}_\delta\tilde{\chi}_i(t),
		\end{system}	
		\]
		by using our delay-free transformation.
		
		From \eqref{eq1}, \eqref{eq2}, and protocol
		\eqref{hoct-comp3-pole-partial}, it means that we only need to prove
		that the system
		\begin{equation}\label{hoct-sys-cl2}
		\begin{system}{ccl}
		\dot{\tilde{x}}(t) &=& A\tilde{x}(t)-\ell_{ii}\beta^{-1} BB\T
		P_{\delta}\tilde{\chi}(t),\\ 
		\dot{\tilde{\chi}}(t) &=& (A+KC)\tilde{\chi}(t)-KC\tilde{x}(t),
		\end{system}
		\end{equation}
		is asymptotically stable for $\alpha\geq\ell_{ii}\geq \beta$.
		
		Define $e(t)=\tilde{x}(t)-\tilde{\chi}(t)$. The system \eqref{hoct-sys-cl2} can be
		rewritten in terms of $x$ and $e$ as
		\begin{equation}\label{hoct-sys-app-2}
		\begin{system}{ccl}
		\dot{\tilde{x}}(t) &=& (A - \ell_{ii}\beta^{-1} B B\T
		P_\delta)\tilde{x}(t)+\ell_{ii}\beta^{-1} BB\T P_\delta e(t) \\ 
		\dot{e}(t) &=& (A+KC+\ell_{ii} \beta^{-1} BB\T P_\delta)e(t)-\ell_{ii}
		\beta^{-1} BB\T P_\delta \tilde{x}(t). 
		\end{system}
		\end{equation}
		Since $\ell_{ii}\geq\beta$, we have
		\begin{multline*}
		(A-\ell_{ii}\beta^{-1} B B\T P_\delta)\T
		P_\delta+P_\delta(A-\ell_{ii} \beta^{-1}BB\T P_\delta) \\
		\leq -\delta I - P_\delta BB\T P_\delta.
		\end{multline*}
		Define $V_1=\tilde{x}\T(t) P_\delta \tilde{x}(t)$ and
		$v=-B\T P_\delta \tilde{x}(t)$. We can derive that 
		\begin{equation*}
		\dot{V}_1 \leq -\delta \|
		\tilde{x}(t)\|^2-\|v\|^2+\theta(\delta)\|e(t)\| \|v\|, 
		\end{equation*}
		where 
		\[
		\theta(\delta)= \ell_{ii} {\beta}^{-1}\left\| B\T P_\delta
		\right\|. 
		\]
		Clearly, $\theta(\delta)\to 0$ as $\delta\to 0$.
		
		Let $Q$ be the positive definite solution of the Lyapunov equation,
		\[
		(A+KC)\T Q+Q(A+KC)=-2I.
		\]
		Since $P_\delta\to 0$ and $\ell_{ii}$ is bounded (we have
		$\ell_{ii}<\alpha$), there exists a $\delta_1$ such that for all
		$\delta \in(0,\delta_1]$,
		\[
		(A+KC+\ell_{ii}\beta^{-1} BB\T P_\delta )\T
		Q+Q(A+KC+\ell_{ii}\beta^{-1} BB\T P_\delta)\leq -I.
		\]
		Define $V_2=e\T(t) Qe(t)$. We get
		\[
		\dot{V}_2 \leq -\|e(t)\|^2+M\|e(t)\| \|v\|
		\]
		where 
		\[
		M= 2 \ell_{ii}{\beta}^{-1}\|QB\|.
		\]
		Define $V=4 M^2 V_1+2V_2$. Then
		\begin{multline*}
		\dot{V}\leq -4 M^2\delta \| \tilde{x}(t)\|^2-2\|e(t)\|^2-4M^2\|v\|^2\\
		+(4M^2\theta(\delta)+2M)\|e(t)\|\|v\| .
		\end{multline*}
		There exists a $\delta^*\leq \delta_1$ such that
		$4M^2\theta(\delta)\leq 2M$ for all $\delta\in(0,\delta^*]$.  Hence
		for a $\delta\in(0,\delta^*]$,
		\[
		\dot{V}\leq -4 M^2\delta \| \tilde{x}(t)\|^2-\|e(t)\|^2-(\|e(t)\|-2M\|v\|)^2.
		\]
		We conclude that the system \eqref{hoct-sys-cl2} is asymptotically
		stable for $\alpha\geq\ell_{ii}\geq \beta$. Based on Lemma
		\ref{thm-tra-part}, the delayed state synchronization result can be
		proved.
	\end{proof}
	
	\section{Delayed output synchronization for heterogeneous MAS with
		communication delays} 
	
	In this section, we consider the following heterogeneous MAS,
	\begin{equation}\label{hemas-eq1}
	\begin{system*}{cl}
	\dot{x}_i(t) &= A_ix_i(t)+B_iu_i(t),\\
	y_i(t) &= C_ix_i(t),
	\end{system*}
	\end{equation}
	where $x_i(t)\in\mathbb{R}^{n_i}$, $u_i(t)\in\mathbb{R}^{m_i}$ and
	$y_i(t)\in\mathbb{R}^p$ are the state, input and the output,
	respectively, of agent $i$ for $i=1,\ldots, N$. Meanwhile, the
	communication network provides agent $i$ with form of \eqref{eq2}
	including time delay $\tau_{ij}$. Similarly, we can obtain \emph{a
		delay-free transformation} for \eqref{hemas-eq1} by letting
	$\tilde{x}_i(t)=x_i(t+\bar{\tau}_{i,1})$,
	$\tilde{y}_i(t)=y_i(t+\bar{\tau}_{i,1})$,
	$\tilde{u}_i(t)=u_i(t+\bar{\tau}_{i,1})$, and
	$\tilde{\zeta}_i(t)=\zeta_i(t+\bar{\tau}_{i,1})$. Meanwhile, MAS
	\eqref{hemas-eq1} satisfies the following assumption.
	
	\begin{assumption}\label{assume2}
		We assume that
		\begin{itemize}
			\item ($A_i$, $B_i$, $C_i$) is stabilizable and detectable.
			\item All eigenvalues of $A_i$ are in the closed left half complex
			plane.
			\item ($A_i$, $B_i$, $C_i$, 0) is right-invertible.
			\item ($A_i$, $B_i$, $C_i$, 0) has no invariant zeros in the closed
			right-half complex plane that coincide with the eigenvalues of
			$A_1$ (the system matrix of the root agent).
		\end{itemize}
	\end{assumption}
	
	Thus, we can transform \eqref{eq1}, \eqref{eq2} and
	\eqref{hocde-state1} as
	\begin{equation}\label{hemas-eq2}
	\begin{system}{ll}
	\dot{\tilde{x}}_i(t)=A_i{\tilde{x}}_i(t)+B_i\tilde{u}(t)\\
	\tilde{y}_i(t)=C_i{\tilde{x}}_i(t)\\
	\tilde{u}_i(t)=\sum_{j=1}^{N}\ell_{ij}F_i\tilde{x}_j(t).
	\end{system}
	\end{equation}
	Given the model \eqref{hemas-eq2} and graph which is a directed
	spanning tree, all earlier approaches can also be applied to a
	heterogeneous MAS. Here we will give
	design scheme to obtain the delayed state synchronization results
	based on \cite{grip-yang-saberi-stoorvogel-automatica}.
	
	Since the graph is equal to a directed spanning tree, it only has a single root
	which is Agent $1$. Moreover, $u_1=0$.  In this section, our  goal is
	achieve delayed output synchronization  
	\begin{equation}\label{osyn}
	\lim_{t\to \infty}\, \left[ y_i(t)-y_j(t-\tau_{ij}) \right] = 0.
	\end{equation}
	
	For the heterogeneous MAS \eqref{hemas-eq1}, we formulate delayed
	output synchronization problem as follows. 	
	\begin{problem}\label{prob3}
		Consider a MAS described by agents \eqref{hemas-eq1} and \eqref{eq2}
		associated with a directed graph
		$\mathcal{G}\in\mathbb{G}_{\beta}^N$ is equal to a spanning tree where $\mathbb{G}_{\beta}^N$
		is defined in Definition \ref{ungrN}. The \emph{delayed output
			synchronization problem} given the set of graph
		$\mathbb{G}_{\beta}^N$ in the presence of unknown, nonuniform,
		arbitrarily large communication delays is to find a distributed
		dynamic protocol of the type \eqref{hocde-state2}, for each agent
		such that \eqref{osyn} is satisfied for all $i\in\{1,\ldots,N\}$,
		for any directed graph $\mathcal{G}\in \mathbb{G}_{\beta}^N$ and for
		any communication delay $\tau_{ij}\in\R^+$.
	\end{problem}
	
	Then, we let 
	\[
	e_i(t)=\tilde{y}_i(t)-\tilde{y}_1(t),
	\]
	and \eqref{hemas-eq2} can be rewritten as
	\begin{equation}\label{hemas-eq3}
	\begin{system*}{rl}
	\begin{pmatrix}
	\dot{\tilde{x}}_i(t)\\\dot{\tilde{x}}_1(t)
	\end{pmatrix}&=\begin{pmatrix}
	A_i&0\\0&A_1\end{pmatrix}\begin{pmatrix}
	\tilde{x}_i(t) \\ \tilde{x}_1(t)
	\end{pmatrix}+\begin{pmatrix}
	B_i \\ 0
	\end{pmatrix} \tilde{u}_i(t)\\
	e_i(t)&=\begin{pmatrix}
	C_i&-C_1
	\end{pmatrix}\begin{pmatrix}
	\tilde{x}_i(t) \\ \tilde{x}_1(t)
	\end{pmatrix}
	\end{system*}
	\end{equation}
	
	Then, let $O_i$ be the observability matrix,
	\[
	O_i=\begin{pmatrix}
	C_i&-C_1\\
	C_iA_i&-C_1A_1\\
	\vdots&\vdots\\
	C_iA_i^{n_i+n_1-1}&-C_1A_1^{n_i+n_1-1}
	\end{pmatrix}.
	\]
	Let $q_i$ denote the dimension of the null space of $O_i$, and define
	$k_i=n_1-q_i$. Next, define $\Lambda_i^u\in\mathbb{R}^{n_i\times q_i}$
	and $\Phi_i^u\in\mathbb{R}^{n_1\times q_i}$ such that
	\[
	O_i\begin{pmatrix}
	\Lambda_i^u \\ \Phi_i^u
	\end{pmatrix}=0,\qquad\rank\begin{pmatrix}
	\Lambda_i^u \\ \Phi_i^u
	\end{pmatrix}=q_i.
	\]
	
	Because ($A_i$, $C_i$) is observable,
	$\Lambda_i^u$ and $\Phi_i^u$ have full column rank. Next, we define
	$\Lambda_i^c$ and $\Phi_i^c$ such that
	$\Lambda_i=[\Lambda_i^u\ \Lambda_i^c]\in\mathbb{R}^{n_i\times n_i}$
	and $\Phi_i=[\Phi_i^u\ \Phi_i^c]\in\mathbb{R}^{n_1\times n_1}$ are
	nonsingular. Thus, we define a new state
	$\bar{x}_i(t)\in\mathbb{R}^{n_i+k_i}$ as
	\[
	\bar{x}_i(t)=\begin{pmatrix}
	\tilde{x}_i(t)-\Lambda_iM_i\Phi_i^{-1}\tilde{x}_1(t)\\-N_i\Phi_i^{-1}\tilde{x}_1(t)
	\end{pmatrix},
	\]
	where
	\[
	M_i=\begin{pmatrix}
	I_{q_i} & 0 \\ 0 & 0
	\end{pmatrix}\in\mathbb{R}^{n_i\times n_1}, \text{ and }N_i=\begin{pmatrix}
	0 & I_{k_i}
	\end{pmatrix}\in\mathbb{R}^{k_i\times n_1}.
	\]
	Based on this new state variable $\bar{x}_i(t)$, we can transform
	\eqref{hemas-eq3} as
	\begin{equation}\label{hemas-eq4}
	\begin{system*}{rl}
	\dot{\bar{x}}_i(t)&=\bar{A}_i\bar{x}_i(t)+\bar{B}_i\tilde{u}_i(t)=\begin{pmatrix}
	A_i&\bar{A}_{12}^i\\0&\bar{A}_{22}^i\end{pmatrix}\bar{x}_i(t)+\begin{pmatrix}
	B_i\\0
	\end{pmatrix}\tilde{u}_i(t)\\
	e_i(t)&=\bar{C}_i\bar{x}_i(t)=\begin{pmatrix}
	C_i&-\bar{C}_{2}^i
	\end{pmatrix}\bar{x}_i(t)
	\end{system*}
	\end{equation} 
	Further, we define ${\phi}_i(t)=\Xi_i\bar{x}_i(t)$ with
	\[
	\Xi_i=\begin{pmatrix}
	\bar{C}_i\\\bar{C}_i\bar{A}_i \\ \vdots \\ \bar{C}_i\bar{A}_i^{\bar{n}-1}
	\end{pmatrix},
	\]
	where $\bar{n}\geq n_i+n_1$. Note that $\Xi_i$ is not necessarily a
	square matrix; however, due to observability of ($\bar{A}_i$,
	$\bar{C}_i$), $\Xi_i$ is injective, which implies that $\Xi_i\T\Xi_i$
	is nonsingular. Meanwhile, we obtain a new expression of
	\eqref{hemas-eq4},
	\begin{equation}\label{hemas-eq5}
	\begin{system*}{l}
	\dot{\phi}_i(t)=(A_o+K_o^i)\phi_i(t)+B_o\tilde{u}_i(t)\\
	e_i(t)=C_o\phi_i(t)
	\end{system*}
	\end{equation}
	where
	\[
	A_o=\begin{pmatrix} 0 & I_{p(\bar{n}-1)} \\ 0 & 0
	\end{pmatrix}\in\mathbb{R}^{p\bar{n}\times p\bar{n}},\ C_o=\begin{pmatrix}
	I_p & 0
	\end{pmatrix}\in\mathbb{R}^{p\times p\bar{n}},
	\]
	$B_o=\Xi_i\bar{B}_i$, and $K_o^i=\begin{pmatrix} 0 \\ G_i
	\end{pmatrix}$ with
	$G_i=\bar{C}_i\bar{A}_i^{\bar{n}}(\Xi\T\Xi)^{-1}\Xi\T$. 
	
	Thus, we will design a dynamic protocol of the form
	\begin{equation}\label{hemas-eq6}
	\begin{system*}{l}
	\dot{\hat{\phi}}_i(t)=(A_o+K_o^i)\hat{\phi}_i(t) + B_o\tilde{u}_i(t) +
	H_\eps Q_\eps C_o\T(\tilde{\zeta}(t) \\
	\hspace*{4.7cm}- \sum_{j=1}^N\ell_{ij}C_o\hat{\phi}_j(t))\\ 
	\hat{\bar{x}}_i(t)=(\Xi\T\Xi)^{-1}\Xi\T\hat{\phi}_i(t)\\
	\tilde{u}_i(t)=F_i\hat{\bar{x}}_i(t)
	\end{system*}
	\end{equation} 
	where
	\[
	F_i=\begin{pmatrix}
	K_i&\Gamma_i-K_i\Pi_i
	\end{pmatrix}
	\]
	$K_i$ is chosen such that $A_i+B_iK_i$ is Hurwitz stable. $\Gamma_i$
	and $\Pi_i$ satisfy the following regulator equations,
	\begin{equation}\label{hemas-eq7}
	\begin{system*}{l}
	\Pi_i\bar{A}_{22}^i=A_i\Pi_i+\bar{A}_{12}^i+B_i\Gamma_i\\
	C_i\Pi_i=\bar{C}_2^i
	\end{system*}
	\end{equation}
	$H_\eps=\diag(I_p\eps^{-1},I_p\eps^{-2},\ldots,
	I_p\eps^{-\bar{n}})\in\mathbb{R}^{p\bar{n}\times p\bar{n}} $,
	$Q_\eps>0$ is the unique solution of the following algebraic Riccati
	equation
	\begin{equation}\label{hemas-eq8}
	(A_o+K_\eps)Q_\eps+Q_\eps(A_o+K_\eps)\T-2\beta Q_\eps C_o\T
	C_oQ_\eps+I_{p\bar{n}}=0 
	\end{equation} 
	with $0<\beta\leq\ell_{ii}$ for $i=2,\ldots N$,
	\[
	K_\eps=\begin{pmatrix}
	0 \\ \eps^{\bar{n}+1}KH_\eps
	\end{pmatrix}\in\mathbb{R}^{p\bar{n}\times p\bar{n}},
	\]
	$K\in\mathbb{R}^{p\times p\bar{n}}$ is chosen matrix. And
	($A_o+K_\eps$, $C_o$) is always observable. Meanwhile, $Q_\eps\to0$
	when $\eps\to0$.
	
	By designing the dynamic protocol \eqref{hemas-eq6}, the following
	theorem can be obtained.
	
	\begin{theorem}\label{hemas-thm}
		Consider a MAS described by agents \eqref{hemas-eq1} and \eqref{eq2}
		satisfying Assumption \ref{assume2}.
		
		The delayed output synchronization problem stated in Problem
		\ref{prob3} is solvable for the set of graphs
		$\mathbb{G}_{\beta}^{N}$. In particular, protocol \eqref{hemas-eq6}
		solves the state synchronization problem for any graph
		$\mathcal{G}\in\mathbb{G}_{\beta}^N$ and any $\tau_{ij}\in \R^+$.
	\end{theorem}
	
	\begin{proof}
		To achieve output regulation our design includes two steps: observer
		design and state feedback design.
		
		Step 1) Observer design.
		
		Let $L_Q$ be the matrix obtained by removing the first row and
		column of $L$ as already used in \eqref{L_Q}. Clearly, $L_Q$ is a
		lower triangular matrix and all its eigenvalues are greater than
		$\beta$, i.e., $\ell_{ii}\geq\beta$.
		
		Let $\tilde{\phi}_i(t)=\phi_i(t)-\hat{\phi}_i(t)$ for $i=2,\ldots,N$, then we
		have
		\begin{align*}
		\dot{\tilde{\phi}}_i(t)&=(A_o+K_o^i)\tilde{\phi}_i(t)-H_\eps Q_\eps
		C_o\T(\tilde{\zeta}(t)-\sum_{j=2}^N\ell_{ij}C_o\hat{\phi}_j(t))\\ 
		&=(A_o+K_o)\tilde{\phi}_i(t)-(K_o-K_o^i)\tilde{\phi}_i(t) \\
		&\hspace*{4cm}-\sum_{j=2}^N\ell_{ij}H_\eps Q_\eps C_o\T C_o\tilde{\phi}_j (t)
		\end{align*}
		with $K_o=[0, K\T]\T$.
		
		Then, we set $\xi_i(t)=\eps^{-1}H_\eps^{-1}\tilde{\phi}_i(t)$, we have
		\begin{equation}
		\eps\dot{\xi}_i(t)=(A_o+K_\eps)\xi_i(t)-K_\eps^i\xi_i(t)-\sum_{j=2}^N\ell_{ij}
		Q_\eps C_o\T C_o\xi_j (t)
		\end{equation}
		where
		\[
		K_\eps^i=\begin{pmatrix} 0\\\eps^{\bar{n}+1}(K_o-K_o^i)H_\eps
		\end{pmatrix}. 	
		\] 
		Since $\bar{L}$ is a positive lower triangular matrix and the graph
		is diverge, we have
		\begin{align*}
		\eps\dot\xi_2(t)&=(A_o+K_\eps-\ell_{22}Q_\eps C_o\T C_o)\xi_2(t)-K_\eps^2\xi_2(t)\\
		\eps\dot\xi_j(t)&=(A_o+K_\eps-\ell_{jj}Q_\eps C_o\T
		C_o)\xi_j(t)-K_\eps^j\xi_j(t)\\
		&\hspace*{4cm}-\ell_{ij}Q_\eps C_o\T C_o\xi_i (t)
		\end{align*}
		with $i<j$. That means we just need to prove
		$A_o+K_\eps-\ell_{ii}Q_\eps C_o\T C_o-K_\eps^i$ is asymptotically
		stable.
		
		Thus, based on \eqref{hemas-eq8} and $\ell_{ii}\geq\beta$, we have
		\begin{align}
		Q_\eps^{-1}(A_o&+K_\eps-\ell_{ii}Q_\eps C_o\T C_o-K_\eps^i)\nonumber\\
		&+(A_o+K_\eps-\ell_{ii}Q_\eps C_o\T C_o-K_\eps^i)\T Q_\eps^{-1}\nonumber\\
		=&Q_\eps^{-1}\left[-I_{p\bar{n}}-2(\ell_{ii} -\beta)Q_\eps C_o\T
		C_oQ_\eps\right.\nonumber\\ 
		&\qquad\qquad\qquad\left.-K_\eps^iQ_\eps-Q_\eps
		(K_\eps^i)\T\right]Q_\eps^{-1}\nonumber\\ 
		\leq&Q_\eps^{-1}\left[-I_{p\bar{n}}-K_\eps^iQ_\eps-Q_\eps
		(K_\eps^i)\T\right]Q_\eps^{-1}. 
		\end{align} 
		Since $K_\eps^i$ is a high-order term of $\eps$, there exist a
		$\eps^*$ such that
		\[
		2\|K_\eps^iQ_\eps\|\leq1.
		\]
		It means we can obtain
		$A_o+K_\eps-\ell_{ii}Q_\eps C_o\T C_o-K_\eps^i$ is asymptotically
		stable for any $\eps<\eps^*$.
		
		Thus, we obtain $\xi_i(t)\to0$ as $t\to\infty$. That is to say, we have
		$\phi_i(t)\to\hat{\phi}_i(t)$. Finally, we have
		\[
		\hat{\bar{x}}_i(t)\to\bar{x}_i(t), \text{ as } t\to\infty.
		\]
		
		Step 2). State feedback designing to solve output regulation.
		
		The observer designing implies $\tilde{u}_i(t)=F_i{\bar{x}}_i(t)$ for
		$t\to\infty$, i.e., control signal converses to a state feedback
		protocol.
		
		Meanwhile, we need to achieve $\lim_{t\to \infty}e_i(t)=0$, it can be
		considered as an output regulation problem. From \eqref{hemas-eq4},
		we define $\bar{x}(t)=[(\bar{x}_1^i(t))\T, (\bar{x}_2^i(t))\T]\T$. Thus, we
		have exosystem
		\[
		\dot{\bar{x}}_2^i(t)=\bar{A}_{22}^i\bar{x}_2^i(t)
		\]
		and the regulation system
		\[
		\dot{\bar{x}}_1^i(t)=A_i\bar{x}_1^i(t)+\bar{A}_{12}^i\bar{x}_2^i(t)+B_i\tilde{u}_i(t).
		\]
		Since ($A_i$, $B_i$) is stabilizable and the eigenvalues of
		$\bar{A}_{22}^i$ are in closed right-half plane, we know protocol
		$\tilde{u}_i(t)=F_i{\bar{x}}_i$ can solve this regulation problem
		based on \cite[Theorem 2.3.1]{saberi-stoorvogel-sannutia}, if the
		regulator equation \eqref{hemas-eq7} is solvable. From
		\cite[Corollary 2.5.1]{saberi-stoorvogel-sannutia}, the regulator
		equations are solvable if, for each $\lambda$ that is an eigenvalue
		of $\bar{A}_{22}^i$, the rank of Rosenbrock system matrix
		\[
		\rank\begin{pmatrix}
		A_i-\lambda&B_i\\C_i&0
		\end{pmatrix}=n_i+p.
		\]
		The Rosenbrock system matrix has normal rank $n_i + p$ due to
		rightinvertibility of the quadruple ($A_i$, $B_i$, $C_i$, 0) (see
		\cite[Property 3.1.6]{h2book}). Since this quadruple has no
		invariant zeros coinciding with eigenvalues of $A_1$ and the
		eigenvalues of $\bar{A}_{22}^i$ are a subset of the eigenvalues of
		$A_1$, it follows that the rank of the Rosenbrock system matrix is
		equal to the normal rank for each $\lambda$ that is an eigenvalue of
		$\bar{A}_{22}^i$.
		
		Thus, we can achieve the delayed output synchronization by
		$\lim_{t\to \infty}e_i(t)=0$, i.e., achieve \eqref{osyn}.
	\end{proof}
	
	\bibliographystyle{plain}
	\bibliography{referenc}
	
\end{document}